\documentclass[a4paper,12pt]{scrartcl}

\usepackage{times}
\usepackage[T1]{fontenc}
\usepackage[dvipsnames]{xcolor}
\usepackage{amsmath, amsfonts, amssymb, amsthm}
\usepackage{mathrsfs}
\usepackage[round]{natbib}
\usepackage{eucal}
\usepackage{nicefrac}
\usepackage{enumitem}
\usepackage{subfigure}
\usepackage{xfrac}
\usepackage{hyperref}
\usepackage{url}

\usepackage{arydshln}
\usepackage{tabularx, booktabs,arydshln}
\makeatletter
\def\adl@drawiv#1#2#3{%
        \hskip.5\tabcolsep
        \xleaders#3{#2.5\@tempdimb #1{1}#2.5\@tempdimb}%
                #2\z@ plus1fil minus1fil\relax
        \hskip.5\tabcolsep}
\newcommand{\cdashlinelr}[1]{%
  \noalign{\vskip\aboverulesep
           \global\let\@dashdrawstore\adl@draw
           \global\let\adl@draw\adl@drawiv}
  \cdashline{#1}
  \noalign{\global\let\adl@draw\@dashdrawstore
           \vskip\belowrulesep}}
\makeatother

\usepackage{setspace}
\newcommand{\cmark}{$+$}
\newcommand{\xmark}{$-$}

\newtheorem{theorem}{Theorem}

\newtheorem{proposition}{Proposition}

\newtheorem{lemma}{Lemma}
\newtheorem{example}{Example}

\theoremstyle{definition}
\newtheorem{definition}{Definition}

\newlength{\wordlength}

\newcommand{\mathwordbox}[3][c]{\settowidth{\wordlength}{$#3$}\text{\makebox[\wordlength][#1]{$#2$}}}

\newcommand{\set}[1]{\{#1\}}

\DeclareMathOperator*{\argmax}{argmax}

\newcommand{\candidates}{C}

\newcommand{\aggscore}{\sigma}

\newcommand{\representative}{\mathit{repr}}
\newcommand{\pluralityscore}{\pi}

\newcommand{\norepr}{r}
\newcommand{\winner}{{j^*}}

\newcommand{\ie}{that is}

\makeatletter
\newtheorem*{rep@theorem}{\rep@title}
\newcommand{\newreptheorem}[2]{%
\newenvironment{rep#1}[1]{%
 \def\rep@title{#2 \ref{##1}}%
 \begin{rep@theorem}}%
 {\end{rep@theorem}}}
\makeatother
\newreptheorem{theorem}{Theorem}
\newreptheorem{proposition}{Proposition}
\newreptheorem{lemma}{Lemma}
\newreptheorem{corollary}{Corollary}

\usepackage{chngcntr}
\usepackage{apptools}
\AtAppendix{\counterwithin{theorem}{section}}
\AtAppendix{\counterwithin{conjecture}{section}}
\AtAppendix{\counterwithin{proposition}{section}}
\AtAppendix{\counterwithin{corollary}{section}}
\AtAppendix{\counterwithin{lemma}{section}}
\AtAppendix{\counterwithin{example}{section}}

\title{A Mathematical Analysis of an\\Election System Proposed\\ by Gottlob Frege}

\renewcommand{\leq}{\le}
\renewcommand{\geq}{\ge}

\author{Paul Harrenstein\thanks{Department of Computer Science, University of Oxford, UK} \and Marie-Louise Lackner\thanks{Christian Doppler Laboratory for Artificial Intelligence and Optimization for Planning and Scheduling, TU Wien, Vienna, Austria} \and Martin Lackner\thanks{Institute of Logic and Computation, TU Wien, Vienna, Austria}}

\date{}

\usepackage{arydshln}

\begin{document}
	
\maketitle

\begin{abstract}
In 1998 a long-lost proposal for an election law by  Gottlob Frege (1848--1925) was rediscovered in the \emph{Th\"uringer Universit\"ats- und Landesbibliothek} in Jena, Germany.
The method that Frege proposed  for the election of representatives of a constituency features a remarkable concern for the representation of minorities. 
Its core idea is that votes cast for unelected candidates are carried over to the next election, while  elected candidates incur a cost of winning. We prove that this sensitivity to past elections guarantees a proportional representation of political opinions in the long run.
We find that through a slight modification of Frege's original method even stronger proportionality
guarantees can be achieved. 
This modified version of Frege's method also provides a novel solution to the apportionment problem, which is distinct from all of the best-known apportionment methods, while still possessing noteworthy proportionality properties.
\end{abstract}

\section{Introduction}

In the summer of~1998 a surprising discovery was made in the \emph{Th\"uringer Universit\"ats- und Landes\-bibliothek} (ThULB) in Jena. 
Hidden among the legacy of the German politician Cle\-mens von Delbr\"uck~(1856-1921), Uwe Dathe, curator at ThULB, found a typescript titled \emph{Vorschl\"age f\"ur ein Wahlgesetz}, which translates to `Proposals for a Voting Law'. 
The author of the typescript turned out to be no one less than Gottlob Frege~(1848--1925), the illustrious logician and recipient of the letter in which Bertrand Russell expounds his famous paradox.
Although the typescript is undated, circumstantial evidence points to~1918 as the almost certain year of composition.
The manuscript was finally published in the original German in the year 2000 \citep{gabriel_dathe:2000a}, accompanied by an extensive and excellent introduction by \citet{dathe-kienzler:2000a}. 
The original typescript has been made digitally available by the \emph{Th\"uringer Universit\"ats- und Landesbibliothek} in Jena~\citep{frege:1918a}.

The discovery of Frege's proposal was surprising and remarkable. 
Regarding it merely as a historical curiosity, however, would not do justice to its originality and perspective.
Rather, we feel that some of its underlying ideas shed a fresh light on modern discussions on how to elect representatives to political assemblies.
In particular, Frege takes a highly original temporal point of view, where the votes cast for unelected candidates in an election are carried over to the next election. He furthermore proposes a rudimentary system of how votes can be delegated from candidate to candidate.
In this paper, we conduct a mathematical analysis of Frege's proposal and can prove that, even from the standpoint of modern social choice and apportionment theory, it fares remarkably well.
In the following, we present Frege's proposal in detail and outline our findings.

\subsection{Structure and content of Frege's proposal}

Frege's typescript consists of four main parts and, in its original form, is presented on~24~numbered printed pages
and is appended by two fold-out tables, \emph{Tafel~I} and \emph{Tafel~II}.
In the following, the page numbering refers to Frege's proposal as published in the year~2000 \citep{frege:2000a}, whereas the bracketed numbers refer to the sheets of the original typescript from~1918 \citep{frege:1918a}.

In a \emph{preliminary note} (\emph{Vorbemerkung}) (page~297, [1]--[3]), Frege sketches the general constitutional provisions of his proposal.
He presumes a division of the electorate in constituencies, with the voters in each constituency delegating an elected representative to an electoral body like the German \emph{Reichstag} (page~297,~[1]). 
This is very much in line with constitutional law in Imperial Germany 
(1871--1918), where the representatives of the constituencies were elected by 
absolute majority and where, in case none of the candidates succeeded in 
securing an absolute majority, a run-off took place between the two candidates 
with the highest number of votes.%
\footnote{The succeeding Weimar Republic (1918--1933) adopted a voting system based on proportional representation, where multiple representatives were elected in each of a smaller number of constituencies of a considerably larger size than was usual in Imperial Germany.}
In this part, he moreover specifies restrictions on active and passive suffrage as he considered them appropriate.

The most original aspects of Frege's election law are to be found in the second part (on pages~297--299, [2]--[4]) of the manuscript. In eleven articles numbered~\S5 through~\S15, Frege presents the \emph{voting method} by means of which the representative of a constituency is chosen.

Finally, Frege makes a short concluding remark (\emph{Schlussbemerkung}, page~299, [4]--[5]) concerning the use of referenda to settle political disputes, which is followed by an extensive discussion elucidating the various aspects of his voting law (\emph{Erl\"auterungen}, pages 299--311, [5]--[24]). The two tables appended by Frege illustrate his voting method by means of two examples he constructed for the purpose (pages~312--313).

\subsection{Political views underlying Frege's proposal}

Frege's proposal for a voting law displays a remarkable juxtaposition of highly conservative and  nationalistic views with classical liberal ideals.
Frege's conservative views are most clearly manifested in his proposal to restrict the right to vote to married men without criminal convictions (``\emph{unbescholten}''), who have performed military service, and who received no state support (``\emph{Almosen}'') in the previous year (197, [1]--[2]).
He explicitly declined women's suffrage on grounds that the husband be the head of the family, the alleged political unit of German society as he saw it (311,~[23]--[24]).
To put this into perspective, it should be noted that women's suffrage was made law on November 30 of the same year that Frege presumably wrote his proposal, with the first German women casting their votes for the \emph{Wahl zur Deutschen Nationalversammlung} on January 19, 1919.

The deep conservatism inherent in these considerations%
\footnote{Later in his life, Frege's conservatism and nationalism would give way to more extreme views on the political far right, in particular a strong antisemitic position, as witnessed by his diary entries from 1924 \citep{gabriel1994gottlob,doi:10.1080/00201749608602425}.
}
is in quite some contrast with the more liberal ideals underlying his {voting method}.
According to Frege himself, the voting method constitutes the fundamental thought of his voting law and can largely be considered independently of his ideas about suffrage (197,~[1]).
In this paper, we will concentrate on the mathematical aspects of Frege's voting method, and leave the other parts of his proposal largely uncommented.

Frege's voting method is based on the classical liberal notion that the \emph{voters} in a constituency should be represented by an elected representative in a national assembly or parliament while upholding the ideal of one-person-one-vote.
Frege's main concern is to guarantee that \emph{no voters' votes are lost in the election process} (197,~[1]; 308--9,~[19]--[20]) by
ensuring that even political minorities should send the constituency's representative from their midst at some point in time.
Failing to do so is a problem that is notoriously inherent in electoral systems where
representatives are elected in single-member constituencies by majoritarian or \emph{first-past-the-post} methods, that is, systems where the voters indicate a single candidate on their ballot, and the candidate occurring on most ballots is elected as representative of the constituency. Arguably, any vote cast on a minority candidate serves the voter, and the candidate, just as well as had the voter not 
participated in the election, and can consequently be considered lost (\emph{Erl\"auterungen}, 303,~[10]).

\subsection{Frege's core idea: elections as a process over time}

Frege's solution to the problem of votes for non-winning candidates being lost during elections
is to acknowledge the fact that elections take place repeatedly over time. 
He sees the election of representatives not so much as a one-shot event, but rather as a series of connected and interdependent events proceeding in rounds.
In his proposal, elections are held every five years (\S6, 298,~[2]), where voters submit a (plurality) ballot, indicating a single candidate, just as it was law in Germany between~1871 and~1918.
He proposes, however, that the representative of a constituency should not be elected on the basis of the votes received in the current election alone, \emph{but also on those cast in previous elections}.
More precisely, every candidate has a \emph{voting score} (``\emph{Stimmenzahl}''), which, after  initially being set to zero (\S5,\S7, 297--8,~[2]), is increased by the number of votes received in each election.
Frege makes a crucial proviso for the incumbent representative
of the constituency: on the Friday before the election, the incumbent's voting score is decreased by the integer part of the  \textit{average voting score} of all candidates (\S14, 299,~[4]). The votes thus subtracted, Frege argues, have served their purpose and cannot be considered
lost (\emph{Vorbemerkung}, 197,~[2]).
We suggest that they can be viewed as the \emph{cost of winning} that the elected representative incurs.
The candidate with the highest voting score after the election is then elected as representative for the next five years.
Possible ties are broken on the basis of age and, in case of an equal number of days lived, by lot (\S13, 299,~[4]).

With Frege's voting method, candidates that only attract minor support among the electorate keep accumulating votes over time, and at some point will be elected as representative.
Thus, also minority opinions are guaranteed to be represented in elected political assemblies, which complies with Frege's guiding principle that no votes should be lost.
Should his voting law be adopted, Frege anticipates the lively participation of all voters in the elections (\emph{Erl\"auterungen}, 303--4,~[11]), a concern that has also attracted attention in social choice theory \citep{fishburn1983paradoxes,moulin1988condorcet}.

To the same end, Frege arranges for the installation of a maximum of twenty-five \emph{choice candidates} (``\emph{Erlesene}'') as those candidates 
in a constituency with maximal voting score (\S8,~298,~[2]).
The representative is chosen among these choice candidates and non-choice candidates need to transfer their votes to one of the choice candidates.
This is to ensure that the votes cast on non-choice candidates are not lost or too much scattered to be of any effect. 
Frege claims that the number of twenty-five choice candidates should suffice for all non-choice candidates to find a politically like-minded choice candidate to transfer their votes to (\emph{Erl\"auterungen},~303,~[10]).
Frege also provisions for the transfer of votes from a choice candidate to a deputy should the former die or otherwise lose his status as a choice candidate (\S11, 298, [3]).
In what follows, however, we will not make this distinction between choice and non-choice candidates.

Given the objectives of Frege's voting method, a key issue is how often a candidate can be expected to be elected as representative given his support in the constituency. 
Even though Frege does not seem to pursue \emph{proportional representation} (``\emph{Proportions\-wahl}'') as a political aim in itself, he states without formal proof that:
``If the strengths of political directions in a constituency remain the same over a prolonged period of time,
then the number of times during which each of these directions is represented 
by the representative of the constituency will behave approximately 
proportionally to these strengths'' \cite[][{302,~[8]}]{frege:2000a}.%
\footnote{The original German reads: ``\emph{Wenn die Richtungen in einem Wahlkreise l\"angere Zeit hindurch die\-selben St\"arken behalten, werden sich die Zeiten, w\"ahrend deren die einzelnen durch den Abgeordneten des Wahlkreises vertreten werden, ann\"ahernd wie diese St\"arken verhalten}''. Translation by the authors.}
The main goal of our paper is to investigate in a mathematically rigorous manner whether Frege's claim can be vindicated.

\subsection{Contributions of this paper}

Frege's intended audience presumably being politicians rather than academics,
Frege neither provided a formal definition nor a formal analysis of his voting method. 
We aim to fill this gap and, in Section~\ref{sec:formulation}, start by formulating his proposal in the precise mathematical framework of modern social choice theory. 
We first observe that Frege's method fails to be \emph{proportional} in the strict sense even under the assumption that the voters' number and preferences remain constant over time:
examples are easily found in which at some point a candidate is elected more often (and another less often) than would be justified by the number of votes he received as a proportion of the total of votes cast (Example~\ref{ex:51111}).
We show, however, that this failure of proportionality can be attributed to the peculiarity of Frege's method that the cost of being chosen representative varies over time. 
Even more, the variation of this cost creates a massive, disproportional advantage for strong candidates who pay a lower cost than minority candidates.

In Section~\ref{sec:prop-orig}, we can nevertheless prove that the cost of winning will converge and stabilize at the number of voters after a finite number of elections--provided that the size of the constituency remains constant (Lemma~\ref{lem:aggaggscore_leq_nm}). 
Frege appears to have been aware of this phenomenon, as he now and then speaks of a stable state (``\emph{Dauerzustand}'') in this context (302,~[8]).
This convergence of the cost of winning can take a very long time, as is also indicated by the two examples Frege himself constructed  to illustrate his procedure (\emph{Erl\"auterungen}, 300--1,~[5]--[7]; \emph{Tafel~I} and \emph{Tafel~II}). 
In both of the examples, the cost of winning can be shown to stabilize only after 184 elections. As Frege proposes that elections are held every five years, this process will take 920 years.

As soon as the cost of winning has stabilized, however, the behaviour of Frege's voting method becomes more favorable.
We can show that, as time proceeds, the proportion of times a candidate is chosen will converge towards the proportion of the candidate's support in the electorate.
In other words, Frege's method achieves proportionality in the limit (Theorem~\ref{thm:frege_asymp}).
Again, this result requires that the size of the electorate remains constant and that the average number of votes each candidate receives converges.
In this context, it is interesting that Frege explicitly expresses the desirability of constituencies being of about equal size and their composition changing as little as possible (\emph{Vorbemerkung}, 197,~[1]). 
If instead the size of the electorate is not assumed to be constant but can grow over time, we give an example showing that this convergence to proportionality does no longer hold (Example~\ref{ex:exp-n}).

The problems that come with the long initialisation phase before the cost of winning stabilizes suggest a modified version of Frege's voting method which we present in Section~\ref{sec:stronger}.
For this modified version, the number of votes cast for each candidate is normalized to lie between~$0$ and~$1$, and the cost of winning is invariably~$1$.
The latter stipulation intuitively corresponds to the cost of winning being equal to the number of voters.
The modified version of Frege's voting method is well-behaved immediately and has stronger proportionality guarantees over time than Frege's original method. 
In particular, we prove that, at any point in time, the number of times each candidate has been chosen lies within a bounded margin from his or her proportional share of votes aggregated up to that time (Theorems~\ref{thm:upperquota} and~\ref{thm:lower-quota}).

Frege's original method was conceived for the election of a single representative in a constituency
and rests on the liberal principle of voting for individual candidates (``\emph{Pers\"onlich\-keits\-wahl}'') instead of a party-list system (``\emph{Listenwahl}'').
That is, the method should guarantee fair representation of citizens' opinions in parliament
rather than reflect the strength of political parties.
Accordingly, it would be inappropriate to represent Frege as making a case for proportional representation as such. 
Nevertheless, the modified Frege method, as proposed in this paper, can naturally be interpreted as an \emph{apportionment method} as is commonly used for assigning seats to parties in a political assembly in systems of proportional representation.
This apportionment method, which we introduce in Section~\ref{sec:apportionment} and refer to as the \emph{Frege's apportionment method}, assigns to each political party as many seats as the number of times it would be elected as representative if the modified Frege method were run for so many times as there are seats in the assembly (parliament) while keeping the electorate fixed.

We can show that Frege's apportionment method is not mathematically equivalent to any of the methods that are common in the literature.
In particular, we demonstrate that it differs from the Adams method,
the D'Hondt (or Jefferson) method, the quota method, the Sainte-Lagu\"e (or Webster) method, the largest remainder method, and the Huntington-Hill method.
Analysing its compliance  with the customary axioms for apportionment, we prove that 
Frege's apportionment method satisfies house monotonicity 
and upper quota but fails population monotonicity. Lower quota is only satisfied if the number of candidates is at most three (Theorem~\ref{thm:apportionment-axioms}).
As an apportionment method, it therefore behaves surprisingly well, especially given the fact that it was not designed as one.
Only the quota method by \citet{balinski1975quota}  satisfies all of the axioms mentioned above that are also satisfied by Frege's method.
The quota method, however, is notoriously biased against small parties. 
Frege's method fares considerably better in this respect, as suggested by a numerical experiment in which we compare the number of votes per representative of the largest party and the number of votes per representative of the smallest party (Section~\ref{subsec:bias}).  
The conclusion seems to be fair that Frege's apportionment method is an interesting and novel addition to the apportionment literature.

Further discussions of Frege's proposals follow at the end of the paper in Section~\ref{sec:disc}. 
So as not to interrupt the flow of the argument, mathematical proofs are deferred to the appendix.
An open-source Python implementation of Frege's voting rule and our modified version is available online [reference omitted for reasons of anonymity].

\subsection{Related work}

Frege's voting method is difficult to compare with other voting rules due to its temporal nature, a feature that voting rules typically do not possess. Thus, we only briefly review some works in social choice theory that combine voting and a temporal structure.
Formalisms such as iterative voting~\citep{Meir2017IterativeVoting} and
dynamic social choice \citep{Tennenholtz2004Transitive,Boutilier2012Dynamic,Parkes2013Dynamic} consider voting scenarios with changing (dynamic) preferences.
In contrast to Frege's proposal, these essentially concern a single election where preferences are updated over time. In particular, they are not concerned with proportional outcomes over time.

Another line of work \citep{conitzer2017fair,Freeman2017Fair,Lackner2020} is concerned with repeating elections, similar to Frege's proposal.
These works, however, focus on fairness towards voters and discuss mechanisms that guarantee a fair distribution of utility among voters over time.
From this point of view, these works can be viewed as an orthogonal approach to the one of Frege, where all emphasis is put on fairness towards candidates and individual voters are not taken into account. We return in Section~\ref{sec:disc} to this partial disregard of voters' current preferences.

Finally, the storable votes method~\citep{casella2005storable,casella2012storable} is a voting rule based on plurality voting.
In each election, voters can decide to either cast a vote or to transfer their vote weight to future elections. If a voter decides to cast a vote, she can spend all, some, or none of her stored weight from previous elections.
This rule vaguely resembles Frege' proposal as it allows minority candidates to win at some point, however only if they have supporters that strategically act to their benefit.
This kind of strategic voting is not required with Frege's proposal; Frege even made his proposal with the intention of reducing strategic voting.

\section{Mathematical Formulation of Frege's Voting Method}\label{sec:formulation}

Frege couched his voting law in legal terms, and also his subsequent discussions are mathematically informal.
In this section we provide a mathematical formulation of Frege's voting law, where we concentrate on the voting mechanism Frege proposed. 
In particular we focus on the way candidates accumulate votes over the course of multiple elections and whether this leads to a fair (proportional) representation of opinions over time.
We will make the simplifying assumption that candidates remain the same over time, 
but we do allow voters to change their number, their identity, and opinion as to their most preferred candidate, unless stated otherwise.
In particular, we will disregard Frege's distinction between choice candidates (``\emph{Erlesene}'') and non-choice candidates, and the delegation mechanism it enables.
Thus, in our analysis, every candidate accumulates votes.

Let~$C$ be a set of $m\ge 2$ \emph{candidates}.
Voting proceeds in rounds over time, meaning that at every point in time $t \ge 1$ an election takes place and a new \emph{representative}~$\representative(t)$ is chosen.
We have~$n_t$ denote the \emph{number of voters} participating in the election at time~$t$.
In every election, the voters submit their preferences by means of \emph{plurality ballots}, \ie, for every $t\ge 1$, the voters specify their most preferred candidate only. 
We denote by~$\pluralityscore_j^t$ the \emph{plurality score} of candidate~$j$ at time~$t$, \ie, the number of voters that put~$j$ on their ballot at time~$t$.
Since every and only participating voters cast votes on candidates, we thus have $n_t = \sum_{j\in C} \pluralityscore_j^t$.
We speak of a \emph{fixed electorate} if the number of voters and the candidates' plurality scores remain constant over time, \ie, if $n_t=n$ and $\pluralityscore_j^t=\pluralityscore_j$ for all $j\in C$ and $t\ge 1$. 

In each round $t$ of the election process, an \textit{aggregate score}~$\aggscore_j^t$ is calculated for every candidate~$j$ on the basis of the scores obtained at the time of the election ($\pluralityscore_j^t$) \emph{and the scores obtained in past elections}. The candidate which obtains the maximal aggregate score at time~$t$ is chosen as representative~\emph{for round~$t$}, that is, $\representative(t)=\argmax_{j\in\candidates}\aggscore_j^t$.
In case of a tie, the candidate that is lexicographically first is chosen. 
This is equivalent to assuming a fixed tie-breaking order, for instance, by breaking ties in favour of the oldest candidate, as suggested by Frege himself.
Formally, we define the \emph{aggregate score~$\aggscore_j^t$} of a candidate~$j$ at time~$t$ inductively such that, for every
 $t\ge 1$,
\begin{align*}
	\aggscore_j^1	&	=	\pluralityscore_j^1
	\\
	\aggscore_j^{t+1}	&	=	
	\begin{cases}
		\rule[-1em]{0pt}{2em}\aggscore_j^t+\pluralityscore_j^{t+1}-\left\lfloor{{\frac{1}{m}\cdot\sum_{k\in\candidates}\aggscore_k^t}}\right\rfloor	&\text{if $\representative(t)=j$,}\\
		\rule[0em]{0pt}{0em}\aggscore_j^t+\pluralityscore_j^{t+1}																&\text{otherwise.}	
	\end{cases}
\end{align*}
The term $\left\lfloor{{\frac{1}{m}\cdot\sum_{k\in\candidates}\aggscore_k^t}}\right\rfloor$ can be seen as the \emph{cost of winning} the election at time~$t$, as it is later subtracted from the aggregate score of the winning candidate. 
Note that this number is chosen in such a way that the aggregate scores of all candidates are guaranteed to remain non-negative at all times.
Furthermore, observe that the aggregate scores at time~$t$ are used to elect the representative at time~$t$ and consequently only include the costs of winning of previous rounds and not of time~$t$.

\begin{example}\label{example:ex1}
Consider a fixed electorate with three candidates and ten voters, and let the corresponding plurality scores of candidates $a$, $b$, and $c$ be~$5$,~$3$, and~$2$,~respectively. 
Table~\ref{tab:ex:1} depicts the values of $\aggscore_j^{t}$ for $t=1, \ldots, 10$. Maximum aggregate scores of each round are printed in bold. 
\begin{table}
{\setlength{\arraycolsep}{6pt}
\[
\begin{array}{crrrcc}
\toprule
\mathrm{time}~t&
\mathwordbox[r]{\aggscore_a^t}{15}	&
\mathwordbox[r]{\aggscore_b^t}{15}	&
\mathwordbox[r]{\aggscore_c^t}{15}	&
\representative(t) & \left\lfloor{{\frac{1}{3}\cdot\sum_{k\in\candidates}\aggscore_k^t}}\right\rfloor
\\\midrule
1 & \mathbf{5} & 3 & 2 & a & 3\\
2 & \mathbf{7} & 6 & 4 & a & 5\\
3 & 7 & \mathbf{9} & 6 & b & 7\\
4 & \mathbf{12} & 5 & 8 & a & 8\\
5 & 9 & 8 & \mathbf{10} & c & 9\\
6 & \mathbf{14} & 11 & 3 & a & 9\\
7 & 10 & \mathbf{14} & 5 & b & 9\\
\cdashlinelr{1-6}
8 & \mathbf{15} & 8 & 7 & a & 10\\
9 & 10 & \mathbf{11} & 9 & b & 10\\
10 & \mathbf{15} & 4 & 11 & a & 10\\
\bottomrule
\end{array}
\]
}
\caption{A simple example of Frege's method (Example~\ref{example:ex1}). Maximum aggregate scores are printed in bold.}\label{tab:ex:1}
\end{table}
At time~$1$, candidate~$a$ is chosen, because~$a$ has a maximum aggregate score of~$5$. 
At time~$2$, each candidate keeps the votes he had obtained at time~$1$ plus the votes obtained at time~$2$, which we assumed to be the same as at~time~$1$. 
Thus for candidates~$b$ and~$c$ the aggregate scores at~time~$2$ are $3+3=6$ and $2+2=4$, respectively.
The cost of winning incurred by candidate~$a$ at time~$1$ amounts to~$3=\lfloor \nicefrac{10}{3}\rfloor$, which has to be subtracted from the number of votes candidate~$a$ received at time~$1$ and time~$2$.
Accordingly, candidate~$a$'s aggregate score $\aggscore_a^2$ at time~$2$ is calculated as $5+5-3=7$.
Hence, at time~$2$, candidate~$a$ again has the highest aggregate score and is elected representative another time.  
The cost of winning at time~$2$, however, has increased to~$5$, and at time~$3$, it is candidate~$b$ who has the maximum aggregate score and is elected representative, this time at a cost of~$7$. And so on.
Note that the cost of winning increases over time, starting with~$3$ and increasing to~$10$, the number of voters.
Once the cost of winning has reached~$10$---indicated in the table by the dashed line---it stabilizes and remains constant at all subsequent time steps.
\end{example}

The increasing cost $\left\lfloor{{\frac{1}{m}\cdot\sum_{k\in\candidates}\aggscore_k^t}}\right\rfloor$ of winning, as we saw in Example~\ref{example:ex1}, suggests an unfairness inherent in Frege's voting method.
Early winners, \ie, those candidates with the highest plurality scores, incur lower costs for being elected than those candidates that win later.
This makes it advantageous to win early in the election process,
creating a positive bias towards strong candidates and accordingly constitutes a disadvantage for minority candidates. 
This is also reflected in these candidates being elected more often than would seem to be justified by the proportion of the electorate that supports them. 
This phenomenon is all the more remarkable as it was Frege's intention to also strengthen minority opinions (\citealt[][{\emph{Erl\"auterungen}, 306,~[15]}]{frege:2000a}; \citealt[][page~292]{gabriel_dathe:2000a}).
How extreme this distortion can be is illustrated by the following example.

\begin{example}
\label{ex:51111}
Let us consider a scenario that highlights the unfairness introduced by increasing costs (Table~\ref{tab:ex:2}). We have a fixed electorate with 6 candidates and 10 voters. The corresponding plurality scores~$\pi_j$ of the candidates~$a$ through~$f$ are $1,1,1,1,1$, and~$5$, respectively. 
The table below depicts the values of $\aggscore_j^{t}$ for $t=1, \ldots, 10$ and $j\in C$.
Ties occur in rows with more than one element printed in bold.
We assume that ties are broken in alphabetical order and thus always to the disadvantage of~$f$.
\begin{table}
{\setlength{\arraycolsep}{6pt}
\[
\begin{array}{crrrrrrcc}
\toprule
\mathrm{time}&
\mathwordbox[r]{a}{15}	&
\mathwordbox[r]{b}{15}	&
\mathwordbox[r]{c}{15}	&
\mathwordbox[r]{d}{15}	&
\mathwordbox[r]{e}{15}	&
\mathwordbox[r]{f}{15}	&
\mathrm{representative} & \left\lfloor{{\frac{1}{6}\cdot\sum_{k\in\candidates}\aggscore_k^t}}\right\rfloor\\\midrule
1 & 1 & 1 & 1 & 1 & 1 & \mathbf{5} & f & 1\\
2 & 2 & 2 & 2 & 2 & 2 & \mathbf{9} & f & 3\\
3 & 3 & 3 & 3 & 3 & 3 & \mathbf{11} & f & 4\\
4 & 4 & 4 & 4 & 4 & 4 & \mathbf{12} & f & 5\\
5 & 5 & 5 & 5 & 5 & 5 & \mathbf{12} & f & 6\\
6 & 6 & 6 & 6 & 6 & 6 & \mathbf{11} & f & 6\\
7 & 7 & 7 & 7 & 7 & 7 & \mathbf{10} & f & 7\\
8 & \mathbf{8} & \mathbf{8} & \mathbf{8} & \mathbf{8} & \mathbf{8} & \mathbf{8} & a & 8\\
9 & 1 & 9 & 9 & 9 & 9 & \mathbf{13} & f & 8\\
10 & 2 & \mathbf{10} & \mathbf{10} & \mathbf{10} & \mathbf{10} & \mathbf{10} & b & 8\\
\bottomrule
\end{array}
\]}
\caption{The unfairness of increasing costs: candidate~$f$ wins unproportionally often (Example~\ref{ex:51111})}\label{tab:ex:2}
\end{table}
After ten rounds, candidate~$f$ has been chosen eight times, candidates~$a$ and~$b$ have been chosen once, and candidates~$c$, $d$ and~$e$ not at all. 
This shows that Frege's voting method does not select representatives in a proportional fashion: in ten rounds it would be possible to perfectly reflect the distribution of votes, \ie, by choosing candidate~$f$ five times and all other candidates once.
\end{example}

At this point, we would like to make a minor remark concerning Frege's use of the floor function in the definition of the cost of winning as $\left\lfloor{{\frac{1}{m}\cdot\sum_{k\in\candidates}\aggscore_k^t}}\right\rfloor$ in his voting method (cf.~$\aggscore_j^{t+1}$, as defined above).\footnote{Frege writes: ``A remainder that is smaller than the number of choice candidates of a constituency, is left disregarded'' (``\emph{Ein Rest, der kleiner als die Anzahl der Erlesenen des Wahlkreises ist, bleibt dabei unber\"ucksichtigt}'') (\S14,~299,~[4], English translation by the authors).}
This was most likely motivated by considerations of numerical simplicity, 
but, for our purposes, does not have a significant mathematical effect.
Replacing the term $\left\lfloor{{\frac{1}{m}\cdot\sum_{k\in\candidates}\aggscore_k^t}}\right\rfloor$ with ${\frac{1}{m}\cdot\sum_{k\in\candidates}\aggscore_k^t}$ does not help to resolve any of the issues identified in Example~\ref{ex:51111}.

In Example~\ref{example:ex1}, we also saw that after time~$7$, the cost of winning stabilizes at~$10$. 
After this time, therefore, candidates no longer gain an advantage by being elected earlier rather than later, at least not with respect to the cost of winning. 
This phenomenon is not specific to Example~\ref{example:ex1};
also in Example~\ref{ex:51111} the cost of winning will stabilize at time~$17$, when it reaches a cost of $10$.
Rather, as long as the \emph{size} of the electorate remains constant, the convergence of the cost of winning to the number of voters will hold generally (Lemma~\ref{lem:aggaggscore_leq_nm}, below). 
Under further conditions on how the opinions in the electorate evolve over time, moreover, the unfairness caused by the varying costs of winning in the initial phase will taper off and result in a proportional representation in the long run. 
We formally prove this in the next section.

The unfairness towards minority candidates caused by the increasing costs of winning suggests a variation of Frege's method where the cost of winning is stipulated to be constant from the outset.
We introduce the \emph{modified Frege method} in Section~\ref{sec:stronger} and show that it not only guarantees proportionality in the long run, but also has stronger proportionality properties, which we will make formally precise.

\section{Proportionality Guarantees for Frege's Method}\label{sec:prop-orig}

In the previous section we saw how the cost of winning converges as long as the number of voters is fixed. We now make this observation mathematically precise.
\newcommand{\lemaggaggscoreleqnm}{Assume that the number of voters is fixed, that is, $n_t=n$ for all $t\ge 1$.
Then the function $a(t)=\sum_{k\in\candidates}\aggscore_k^t$ is monotonically increasing. Moreover, there exists a positive number $t_0$ such that $a(t)=n\cdot m$ for all $t\ge t_0$.}
The proof of this lemma as well as all further proofs can be found in the appendix.
\begin{lemma}
\lemaggaggscoreleqnm\label{lem:aggaggscore_leq_nm}
\end{lemma}

This convergence process of the cost of winning (which is $\lfloor \frac{a(t)}{m} \rfloor$) to the number of voters can take quite a long time:
In Frege's proposal, he provided two explanatory examples with $n=1000$ and $m=25$.
With these parameters, Frege's voting method reaches a constant cost of winning at time
$t_0=184$. As Frege proposes that elections are held every five years, this would amount to 920 years.

Let us continue by specifying in which sense Frege's voting method violates even a most basic form of proportionality, as noted in Example~\ref{ex:51111}.
In what follows, let $\rho_j(t)$ denote the number of times candidate~$j$ is chosen as representative up until time~$t$, \ie, $\rho_j(t)=|\left\{s \le t: \representative (s)=j\right\}|$.
Under the assumption of a fixed electorate, after $t$ rounds, each candidate~$j$ should ideally win $t\cdot\frac{\pluralityscore_j}{n}$ times, \ie, the number of rounds multiplied by the proportion of the electorate that supports candidate~$j$.
If $t\cdot\frac{\pluralityscore_j}{n}$ is an integer for all $j$, a perfectly proportional outcome is possible.
This observation gives rise to the following definition.

\begin{definition}\label{def:integral-quota}
A fixed electorate with plurality scores $(\pluralityscore_j)_{j\in C}$ has \emph{integral quotas at time~$t$} if $t\cdot\frac{\pluralityscore_j}{n}$  is integral for all candidates~$j$.
We say that an (infinite) sequence of chosen representatives $(\representative(1),\representative(2),\dots)$ satisfies \emph{variable integral quota} if for any time~$t\ge 1$ at which the electorate has integral quotas,
it holds that 
\begin{align*}
\rho_j(t) = t\cdot\frac{\pluralityscore_j}{n}
\quad\text{ for every $j\in C$}.
\end{align*}
\end{definition}

Example~\ref{ex:51111} shows that Frege's method cannot guarantee this property: a sequence of representatives chosen by Frege's method may violate variable integral quota.
Note that variable integral quota applies to rather few electorates, namely only fixed electorates considered at time points when a perfectly proportional outcomes is possible.
Hence, we consider variable integral quota as a weak and very basic form of proportionality.
However, the following theorem shows under which conditions $\frac{\rho_j(t)}{t}$ will converge to $\frac{\pluralityscore_j}{n}$. Thus, Frege's method provides a form of proportionality \emph{in the long run}.

\newcommand{\thmfregeasymp}{If we assume a fixed electorate,  that is, $n_t = n$ and $\pluralityscore_j^t = \pluralityscore_j$ for all $t\ge 1$, the following holds for Frege's voting method:
\[
\lim_{t \to \infty} \frac{\rho_j(t)}{t} = \frac{\pluralityscore_j}{n}.
\]
If $n_t = n$ and for all $t\ge 1$ and for all $j \in C$ there is some $\pluralityscore_j^* \in [0,1]$ such that
\[ 
\lim_{t \to \infty} \frac{\sum_{s=1}^t\pluralityscore_j^s}{t} = \pluralityscore_j^*, 
\]
then the following holds for all candidates $j\in C$:
\[
\lim_{t \to \infty} \frac{\rho_j(t)}{t} = \frac{\pluralityscore_j^*}{n}.
\]}

\begin{theorem}
\thmfregeasymp
\label{thm:frege_asymp}
\end{theorem}

The following proposition shows that, for a fixed electorate, proportionality is not only guaranteed in the limit but eventually also within (finite) intervals.
\newcommand{\propperiodic}{If we assume a fixed electorate with $n$ voters, there exists a time $t^* \ge 1$ and a period length $P \in \mathbb{N}$ such that, for all $t \ge t^*$  and all $j \in C$,
\[
\frac{\rho_j(t + P) - \rho_j(t)}{P} = \frac{\pi_j}{n}\text.
\]}
\begin{proposition}
\propperiodic
\label{prop:periodic}
\end{proposition}
If we do not assume a fixed number of voters, proportionality cannot be guaranteed even as $t\to\infty$. 
To see this, consider the following example.

\begin{example}\label{ex:exp-n}
We consider a scenario with two candidates, $a$ and $b$. For $t=1$, there are three voters and $\pi_a^1=2$, $\pi_b^1=1$. In every following round the number of voters is doubled and the ratio 
${\pi_a^t}/{\pi_b^t}$
 remains the same, namely~$2$.
We thus have $\pi_a^t= 2^{t}$ and $\pi_b^t=2^{t-1}$. 
As we will see, candidate~$b$ never wins despite receiving one third of the votes.

In order to see this, let us prove that $\aggscore_a^t > \aggscore_b^t$ for all $t \ge 1$ by induction over $t$. The basis is clearly fulfilled since $\aggscore_a^1 = \pi_a^1=2 > 1 = \pi_b^1 =\aggscore_b^1$.
For the induction step, we assume that $\aggscore_a^s > \aggscore_b^s$ for all $s \le t$, \ie, candidate $a$ has always won so far. We thus have:
\begin{align*}
\aggscore_a^{t+1} & = \aggscore_a^t + 2^{t+1} - \left\lfloor\frac{\aggscore_a^t + \aggscore_b^t}{2} \right\rfloor \quad\text{ and }\quad
\aggscore_b^{t+1} = \aggscore_b^t + 2^{t}.
\end{align*}
Since $\aggscore_a^t > \aggscore_b^t$, it follows:
\[
\left\lfloor\frac{\aggscore_a^t + \aggscore_b^t}{2} \right\rfloor \le \frac{\aggscore_a^t + \aggscore_b^t}{2} <  \aggscore_a^t.
\]
Moreover, since $b$ has not been chosen so far,
\[
\aggscore_b^{t+1} = \sum_{s=1}^{t+1} \pi_b^s = \sum_{s=1}^{t+1} 2^{s-1} = 2^{t+1} -1.
\]
We thus have:
\begin{align*}
\aggscore_a^{t+1} & = \aggscore_a^t + 2^{t+1} - \left\lfloor\frac{\aggscore_a^t + \aggscore_b^t}{2} \right\rfloor 
 > \aggscore_a^t + 2^{t+1} - \aggscore_a^t = 2^{t+1}
 > \aggscore_b^{t+1}.
\end{align*}
Thus, candidate $b$ will never be chosen as representative.
\end{example}

In conclusion, Frege's method is not proportional for arbitrary time intervals, but converges to proportional outcomes if the size of the electorate is fixed.

\section{The Modified Frege Method}\label{sec:stronger}

The examples in Section~\ref{sec:formulation} and the results of Section~\ref{sec:prop-orig} point to an increasing cost of winning as the reason for Frege's original method failing a reasonable form of proportionality in the initial phase before the cost of winning has stabilized. 
This observation suggests that a natural variation of Frege's original method, for which the cost of winning is stipulated to be constant, might do better.
We thus introduce the following modification of Frege's method,  which we will refer to as the \emph{modified Frege method}.
In the formal definition of the modified Frege method, we abstract from the size of the electorate and accordingly use \emph{normalized plurality scores}~$p_j^t$ for candidates $j$ and times $t\ge 1$: 
\[
p_j^t = \frac{\pluralityscore_j^t}{n_t}\text,
\]
where~$n_t$ denotes the total number of voters at time~$t$.
The aggregate scores for the modified Frege method are defined as follows, where we use Latin letters to denote variables instead of Greek ones as in the definition of Frege's original method:
\begin{align*}
	s_j^1	
	&	=	p_j^1
	\\
	s_j^{t+1}	&	=	
	\begin{cases}
		\rule[-1em]{0pt}{2em}
		s_j^t+p_j^{t+1}-1	&\text{if $\representative(t)=j$,}\\
		s_j^t+p_j^{t+1}																&\text{otherwise.}
	\end{cases}
\end{align*}

One may wonder why the number~$1$ is subtracted from the winning candidate. The reason is that this choice ensures that the sum $\sum_{j\in C} s_j^t$ of aggregated scores is invariably~$1$ for all times~$t$.
Hence, this stipulation intuitively corresponds to the cost of winning being equal to the number of voters, as it is eventually the case for Frege's original method (cf.~Lemma~\ref{lem:aggaggscore_leq_nm}).
A potential disadvantage of the modified Frege method is that aggregated scores may become negative; we will further discuss this issue at the end of the paper in Section~\ref{sec:disc}.

\begin{table}[t]
{\setlength{\arraycolsep}{6pt}
\[
\begin{array}{crrrrrrc}
\toprule
\mathrm{time}&
\mathwordbox{a}{15}	&
\mathwordbox{b}{15}	&
\mathwordbox{c}{15}	&
\mathwordbox{d}{15}	&
\mathwordbox{e}{15}	&
\mathwordbox{f}{15}	&
\mathrm{representative} 
\\\midrule
1 &  0.1 & 0.1 & 0.1 & 0.1 & 0.1 & \mathbf{0.5} &f\\
2 &  \mathbf{0.2} & \mathbf{0.2} & \mathbf{0.2} & \mathbf{0.2} & \mathbf{0.2} & 0.0 & a\\
3 &  -0.7 & 0.3 & 0.3 & 0.3 & 0.3 & \mathbf{0.5} & f\\
4 &  -0.6 & \mathbf{0.4} & \mathbf{0.4} & \mathbf{0.4} & \mathbf{0.4} & 0.0 &b\\
5 &  -0.5 & -0.5 & \mathbf{0.5} & \mathbf{0.5} &\mathbf{0.5} & \mathbf{0.5} &  c\\
6 & -0.4 & -0.4 & -0.4 & 0.6 & 0.6 & \mathbf{1.0} & f\\
7 &  -0.3 & -0.3 & -0.3 & \mathbf{0.7} & \mathbf{0.7} & 0.5 & d\\
8 &  -0.2 & -0.2 & -0.2 & -0.2 & 0.8 & \mathbf{1.0} & f\\
9 & -0.1 & -0.1 & -0.1 & -0.1 & \mathbf{0.9} & 0.5 & e\\
10 &  0.0 & 0.0 & 0.0 & 0.0 & 0.0 & \mathbf{1.0} & f\\
\bottomrule
\end{array}
\]
}
\caption{Example illustrating the modified Frege method (Example~\ref{example:ex_modiefied})}
\label{table:ex_modified}
\end{table}

\begin{example}\label{example:ex_modiefied}
Let us now reconsider Example~\ref{ex:51111} for the  modified Frege method, as depicted in Table~\ref{table:ex_modified}.
The normalized plurality scores for candidates~$a$ through~$f$ are $0.1$, $0.1$, $0.1$, $0.1$, $0.1$, and~$0.5$, respectively.
We  now see that this method produces a proportional outcome: candidate~$f$ wins five times and all other candidates once, which is exactly in accordance with the candidates' proportional share of the votes.
\end{example}

The modified Frege method enjoys a number of proportionality properties that are stronger than the ones that can be proven for Frege's original method.
In the following, let  $r_j(t)$ denote the number of times candidate~$j$ is chosen as representative up to time~$t$, that is $r_j(t)=|\left\{s \le t: \representative (s)=j\right\}|$. 
First, we show that a similar statement to Theorem~\ref{thm:frege_asymp} holds for the  modified Frege method:
\newcommand{\thmmodfregeasymp}{If we assume 
 that $p_j^t=p_j$ for all $t \geq 1$, the following holds for the modified Frege method:  
\[
\lim_{t \to \infty} \frac{r_j(t)}{t} = p_j.
\]
If the normalized plurality scores $p_j^t$ are not fixed but for all $j \in C$ there is some $p_j^* \in [0,1]$ such that
\[ 
\lim_{t \to \infty} \frac{\sum_{s=1}^tp_j^s}{t} = p_j^*\text,
\]
then the following holds for all candidates $j\in C$:
\[
\lim_{t \to \infty} \frac{r_j(t)}{t} = \lim_{t \to \infty} \frac{\sum_{s=1}^t p_j^s}{t} = p_j^*.
\]}
\begin{theorem}
\thmmodfregeasymp\label{thm:mod_frege_asymp}
\end{theorem}

Theorem~\ref{thm:mod_frege_asymp} provides a proportionality guarantee for the modified Frege method in the long run. 
Note that Theorem~\ref{thm:mod_frege_asymp} does not require the number of voters to be fixed, in contrast to the analogous result for Frege's original method where this assumption is necessary (cf.\ Theorem~\ref{thm:frege_asymp} and Example~\ref{ex:exp-n}).
We will now aim for much stronger guarantees, namely guarantees that hold for arbitrary time intervals.
We strengthen the definition of variable integral quota (Definition~\ref{def:integral-quota}) to hold for arbitrary electorates, inspired by the lower and upper quota axioms in the apportionment setting (cf.\ Section~\ref{sec:apportionment}).

\begin{definition}
For all candidates $j\in C$, let $p_j^1, p_j^2, \dots$ be an infinite sequence of normalized plurality scores.
We say that an (infinite) sequence of chosen representatives $\representative(1)$, $\representative(2)$, $\dots$ satisfies \emph{variable upper quota} if for any time~$t\ge 1$ 
it holds that 
\[r_j(t) \le \left\lceil\sum_{s=1}^{t}p_j^s\right\rceil
\quad\text{ for every }j\in C,
\]
and it satisfies \emph{variable lower quota} if for any time~$t\ge 1$ 
it holds that 
\[r_j(t) \ge \left\lfloor\sum_{s=1}^{t}p_j^s\right\rfloor\quad\text{ for every }j\in C.\]
\end{definition}

Note that both variable upper and lower quota imply variable integral quota: in case of integral quotas any deviation from a proportional distribution would also violate variable upper and lower quota. 
In the following we say that the modified Frege method satisfies variable lower or upper quota if any sequence of winners produced by this method satisfies the corresponding axiom.

\newcommand{\thmupperquota}{The modified Frege method satisfies variable upper quota.}

\begin{theorem}
\thmupperquota\label{thm:upperquota}
\end{theorem}

As a consequence, the modified Frege method also satisfies variable integral quota. By contrast, it violates variable lower quota, as the following example illustrates.

\begin{example}
\label{ex:lowerquota}
Let us consider a fixed electorate with six candidates and $2750$ voters. The plurality scores are $1001,$ $1000$, $206$, $182$, $181$, and~$180$, respectively. For increased readability, the corresponding normalized plurality scores are obtained by dividing these numbers by~$2750$ in
Table~\ref{table:lowerquota}. 
\begin{table}
{\setlength{\arraycolsep}{6pt}
\[
\begin{array}{crrrrrrc}\toprule
\mathrm{time}&
\mathwordbox{a}{15.0}	&
\mathwordbox{b}{15.0}	&
\mathwordbox{c}{15.0}	&
\mathwordbox{d}{15.0}	& 
\mathwordbox{e}{15.0}	& 
\mathwordbox{f}{15.0}	& 
\mathrm{representative}\\\midrule
1 & \mathbf{1001} & 1000 & 206 & 182 & 181 & 180 & a\\
2 & -748 & \mathbf{2000} & 412 & 364 & 362 & 360 & b\\
3 & 253 & 250 & \mathbf{618} & 546 & 543 & 540 & c\\
4 & \mathbf{1254} & 1250 & -1926 & 728 & 724 & 720 & a\\
5 & -495 & \mathbf{2250} & -1720 & 910 & 905 & 900 & b\\
6 & 506 & 500 & -1514 & \mathbf{1092} & 1086 & 1080 & d\\
7 & \mathbf{1507} & 1500 & -1308 & -1476 & 1267 & 1260 & a\\
8 & -242 & \mathbf{2500} & -1102 & -1294 & 1448 & 1440 & b\\
9 & 759 & 750 & -896 & -1112 & \mathbf{1629} & 1620 & e\\
10 & 1760 & 1750 & -690 & -930 & -940 & \mathbf{1800} & f\\
11 & \mathbf{2761} & 2750 & -484 & -748 & -759 & -770 & a\\
\bottomrule
\end{array}
\]
}
\caption{The modified Frege method violates lower quota (Example~\ref{ex:lowerquota}). For increased readability, the aggregate scores~$s_j^t$ are multiplied by~$2750$.}
\label{table:lowerquota}
\end{table}
The variable lower and upper quota of candidate~$b$ at round $11$ is $\frac{11\cdot 1000}{2750}=4$, but candidate~$b$ has been chosen only~$3$ times, \ie, $r_b(11)=3$. 
Similar examples can be found for $m=4$ (for instance, with plurality scores of $1001, 1000, 115$, and~$26$, and for $t=30$) and $m=5$ (for instance, with plurality scores of $1001, 1000, 300, 107$, and~$92$ and for $t=15$).
\end{example}

We can nevertheless show that the violations of variable lower quota by the  modified Frege method are not too severe.
This is in particular the case for electorates with few candidates, as the following theorem shows.

\newcommand{\thmlowerquota}{For $m \in \{2,3\}$, the modified Frege method satisfies variable lower quota.
For $m \ge 4$, we have $\norepr_j(t) \ge \left\lfloor\sum_{s=1}^{t}p_j^s\right\rfloor- \left\lceil\frac{m-3}{2}\right\rceil$ for every candidate $j$ and time $t\ge 1$.}
\begin{theorem}
\thmlowerquota\label{thm:lower-quota}
\end{theorem}

The following example shows, however, that variable lower quota violations can still be arbitrarily large. 
Yet, this construction requires a number of candidates that is exponential in the size of the violation.
Thus, it may still be possible to strengthen the bounds of Theorem~\ref{thm:lower-quota} for the cases in which $m\ge 6$. (Example~\ref{ex:lowerquota} shows that Theorem~\ref{thm:lower-quota} is optimal for $m=4$ and $m=5$.)

\begin{example}
We define a variable electorate with candidates $C=\{1,\dots,m\}$ and $n_t = m-t+1$ voters for $t\in \{1,\dots,m\}$.
At time $t$, we have plurality scores of 
\begin{align*}
& \pi_j^t = 0 &&\text{for }j\in \{1,\dots,t-1\},\\
& \pi_j^t = 1 &&\text{for }j\in \{t,\dots,m\}.
\end{align*}
The corresponding normalized plurality scores are
\begin{align*}
& p_j^t = 0 				&&\text{for }j\in \{1,\dots,t-1\},\\
& p_j^t = \frac{1}{m-t+1} 	&&\text{for }j\in \{t,\dots,m\}.
\end{align*}
Furthermore, we assume that if candidate~$i$ and~$j$ are tied, then the tie is broken in favor of~$\min(i,j)$.
Due to this tie-breaking assumption, the modified Frege method selects candidate~$1$ in the first round, candidate~$2$ in the second, candidate~$t$ in round~$t$.
Let us consider round~$m-1$, in which candidate~$m-1$ wins. 
The variable lower quota of candidate~$m$ is \[
\left\lfloor\sum_{s=1}^{m-1}p_m^s\right\rfloor=\left\lfloor\frac{1}{m}+\frac{1}{m-1}+\dots+\frac{1}{2}\right\rfloor=\left\lfloor\sum_{i=2}^{m}\frac{1}{i}\right\rfloor> \sum_{i=1}^{m}\frac{1}{i}-2.
\] 
Since the harmonic series $\sum_{i=1}^\infty \frac 1 i$ grows without limit, the variable lower quota of candidate~$m$ is unbounded for a growing number of candidates~($m$).
Recall that candidate~$m$ does not win before round~$m$.
Thus, if $m$ tends to infinity, so does the violation of candidate $m$'s variable lower quota at time $t=m-1$.
\end{example}

\section{The Apportionment Setting}
\label{sec:apportionment}

In this section, we want to analyse Frege's methods from the viewpoint of apportionment.
Let us first review the apportionment problem and well-known methods that provide apportionment solutions.

\subsection{Apportionment methods}

An \emph{apportionment problem} for $m$ parties is given by a distribution $\mathbf p = (p_1,\dots,p_m)$ of votes with $\sum_{i=1}^m p_i = 1$ and a desired house size $k$. A \emph{solution} to the apportionment problem $(\mathbf p,k)$ is an $m$-sequence of non-negative integers $(a_1,\dots,a_m)$ with $\sum_{i=1}^m a_i = k$.
An \emph{apportionment method} is a function that returns for every apportionment problem a valid solution\footnote{To simplify the presentation, we assume that ties are broken in some fashion and thus apportionment methods always return a single solution.}.
Apportionment has two main applications: to assign a fixed number of parliamentary seats to parties (proportionally to their vote count), and to assign representatives in a senate to states (proportionally to their population count). From a mathematical point of view, these two applications are indistinguishable, and, in particular, this distinction is not relevant for our study.
For the sake of clarity, we speak in the following of parties and seats.
First, we are going to introduce some important apportionment methods \citep[cf.][]{Balinski1982FairRepresentation:Meeting,Pukelsheim2017ProportionalRepresentation:Apportionment}.

\paragraph{Largest remainder method.} The earliest proposal for an apportionment method is the \emph{largest remainder method} (or Hamilton method).
The largest remainder method assigns in a first step $\lfloor kp_i \rfloor$ seats to each party. In a second step, all remaining seats are distributed so that each party receives at most one seat. Priority is given to parties with the largest remainder, that is, those with largest $kp_i - \lfloor kp_i \rfloor$.

\paragraph{Divisor methods.} Divisor methods are the most commonly used apportionment methods. Their definition is based on \emph{divisor criteria}: A divisor criterion is a  monotonically increasing function $d: \mathbb{N}\to\mathbb{R}$ that satisfies $i \le d(i) \le i + 1$ for all $i\geq 0$.
A divisor criterion $d$ induces a $d$-rounding, defined as follows: \[\left[x\right]_d = \left\{a\in \mathbb{N} : d(a-1)\le x \le d(a)\right\}.\] If $x=d(a)$ for some $a$, then $\left[x\right]_d$ contains two integers, otherwise only one.
For example, the divisor criterion $d(a)=a+1$ corresponds to rounding down, with the slight difference that rounding down an integer~$x$ with $x=a$ yields here both $a$ and $a-1$.  

Given a divisor criterion $d$ we define a corresponding divisor method: the set of $d$-admissible solutions is defined as \[\left\{(a_1,\dots,a_m) \in \mathbb{N}^m: \sum_{i=1}^m a_i = k \text{ and }a_i \in \left[\frac{p_i}{x}\right]_d \text{ for some positive }x\in \mathbb{R}\right\}.\]
As we require that apportionment methods return only one solution, a tie-breaking mechanism may be necessary to choose one solution in this set.

We can now define the most common divisor methods:
the D'Hondt method (or Jefferson method) is defined by $d(a)=a+1$, that is, rounding down. The Adams method is defined by $d(a)=a$ (rounding up). The  Sainte-Lagu\"e method (or Webster method) is defined by $d(a)=a+0.5$, which corresponds to rounding to the nearest integer. 
Finally, the Huntington-Hill method uses the $d(a)=\sqrt{a(a+1)}$ criterion.

\paragraph{Quota method.} The quota method~\citep{balinski1975quota} is the most recent addition to this list of apportionment methods. It is defined iteratively, starting with the empty solution $(0,\dots,0)$.
In round $\ell\geq 1$, if $(a_1,\dots,a_m)$ is the current solution ($\sum a_i = \ell-1$), we consider all parties that would not violate upper quota (cf. Definition~\ref{def:uppper_quota}) if they received an additional seat, that is, all parties $i$ with $a_i + 1\leq \lceil p_i \ell \rceil$ or, equivalently, $a_i < p_i \ell$.
Then we choose among these parties the one party $i$ with maximum $p_i/(a_i +1)$ 
(subject to a tie-breaking, if necessary); party~$i$ receives another seat.

\paragraph{Frege's apportionment method.} Both Frege's original method and
the modified Frege method can easily be transformed into apportionment methods.
However, since Frege's original method violates even a very basic proportionality property (weak proportionality, see below),  
it is not a sensible method in this context and we omit it here from further study.
To apply the modified Frege method, we
view the vote distribution $(p_1,\dots,p_m)$ as a fixed electorate and apply the method for $k$ rounds, thus obtaining an apportionment solution $(r_1(k),\dots,r_m(k))$.
Let us refer to this method as \emph{Frege's apportionment method}.

Given this interpretation, it is natural to ask how Frege's apportionment method compares to other apportionment methods, and, in particular, whether it is equivalent to an already established method in the apportionment setting. Example~\ref{ex:frege_distinct} gives a negative answer to this question. 

\begin{example}\label{ex:frege_distinct}
A concrete example where all apportionment methods listed in Table~\ref{tab:app} yield different solutions is given by 
$\mathbf p = (\frac{79}{98}, \frac{7}{98}, \frac{6}{98}, \frac{3}{98}, \frac{2}{98}, \frac{1}{98})$ 
and a house size $k=20$.
We omit the calculations and only list the results:\medskip

\[
\begin{array}{l@{\quad}l}
\text{Largest Remainder:} & (16, 2, 1, 1, 0, 0) \\
\text{D'Hondt (Jefferson):} & (18, 1, 1, 0, 0, 0)\\
\text{Adams:} & (14, 2, 1, 1, 1, 1)\\
\text{Sainte-Lagu\"e (Webster):} & (17, 1, 1, 1, 0, 0) \\
\text{Huntington-Hill:} & (15, 1, 1, 1, 1, 1)\\
\text{Quota method:} & (17, 2, 1, 0, 0, 0)\\
\text{Frege's apportionment method:} & (16, 1, 1, 1, 1, 0)
\end{array}
\]
\end{example}
In the following section, we address the issue in more depth from an axiomatic perspective.

\subsection{Apportionment axioms}
\label{subsec:frege-app}

For an overview of apportionment methods and their respective properties, we refer the reader to Table~\ref{tab:app}; the corresponding analysis can be found, e.g., in the book by \cite{Balinski1982FairRepresentation:Meeting}. In the following, we discuss axiomatic properties of 
Frege's apportionment method.

{\setlength{\tabcolsep}{3pt}
\newcolumntype{P}[1]{>{\centering\arraybackslash}p{#1}}
\begin{table}
\begin{center}
\begin{tabular}{lP{4.0em}P{4.0em}P{3.7em}P{3.7em}P{4.2em}}
\toprule  & house monot. & popul. monot. & lower quota & upper quota & quota for $m=3$  \\\midrule
Largest Remainder & \xmark & \xmark & \cmark & \cmark & \cmark  \\ 
D'Hondt (Jefferson) & \cmark & \cmark & \cmark & \xmark & \xmark \\ 
Adams & \cmark & \cmark & \xmark & \cmark & \xmark \\ 
Sainte-Lagu\"e (Webster) & \cmark & \cmark & \xmark & \xmark & \cmark \\ 
Huntington-Hill & \cmark & \cmark & \xmark & \xmark & \xmark \\ 
Quota method & \cmark & \xmark & \cmark & \cmark & \cmark \\ 
Frege's apportionment method & \cmark & \xmark & \xmark & \cmark & \cmark \\ \bottomrule
\end{tabular} 
\end{center}
\caption{An overview of apportionment methods and their respective properties.
}
\label{tab:app}
\end{table}
}

As a first step, we want to discuss a basic requirement of apportionment methods, called weak proportionality.
\begin{definition}
An apportionment method satisfies \emph{weak proportionality} if,
given an apportionment problem $((p_1,\dots,p_m), k)$ with $k\cdot p_i$ being integer for every $i\in \{1,\dots,m\}$, the method returns $(kp_1,\dots,kp_m)$.
\label{def:weak_proport}
\end{definition}
It is easy to see that our concept of variable integral quota is closely related to weak proportionality.
Thus, Frege's original method, seen as an apportionment method, violates this property.
This is the reason why we focus on Frege's apportionment method (which is based on the modified Frege method). Let us now consider two stronger proportionality requirements:

\begin{definition}
An apportionment method satisfies \emph{upper quota} if, for any apportionment problem $((p_1,\dots,p_m), k)$, the method returns a solution $(a_1,\dots,a_m)$ satisfying $a_i\le \lceil kp_i\rceil$ for all $i\in \{1,\dots,m\}$.
An apportionment method satisfies \emph{lower quota} if, for any apportionment problem $((p_1,\dots,p_m), k)$, the method returns a solution $(a_1,\dots,a_m)$ satisfying $a_i\geq \lfloor kp_i\rfloor$ for all $i\in \{1,\dots,m\}$.
An apportionment method satisfies \emph{quota} if it satisfies both lower and upper quota.
\label{def:uppper_quota}
\end{definition}

Note that upper quota implies weak proportionality since any deviation from the proportional solution $(kp_1,\dots,kp_m)$ would violate upper quota for some voter.
The same holds for lower quota.

Frege's apportionment method satisfies upper quota (and thus weak proportionality) but fails lower quota.
This is an immediate consequence of Theorem~\ref{thm:upperquota} and Example~\ref{ex:lowerquota}, respectively.
Note that Theorem~\ref{thm:lower-quota} also holds in the apportionment setting and thus Frege's apportionment method satisfies quota for $m\in\{2,3\}$, and for $m\ge 4$ violates lower quota by at most $\lceil\frac{m-3}{2}\rceil$.

Let us now turn to two monotonicity axioms, viz., house and population monotonicity.

\begin{definition}
An apportionment method satisfies \emph{house monotonicity} if the following holds:
for any vote distribution $\mathbf p$ and positive integer $k$, if this method returns the solution $(a_1$, $\dots$, $a_m)$ for the problem $(\mathbf p,k )$ and the solution $(b_1,\dots,b_m)$ for the problem $(\mathbf p,k+1)$, then there exists $1\le i \le m$ such that (i) $a_i+1 = b_i$ and (ii) $a_j=b_j$ for all $j\neq i$.
\end{definition}

In other words, if the house size increases by one, then the apportionment solution can change only by an increase of $1$ for one party. The largest remainder method is notable in that it actually violates this basic criterium.
As Frege's apportionment method is calculated iteratively, it is easy to see that it satisfies house monotonicity.

\begin{definition}
An apportionment method satisfies \emph{population monotonicity} if the following holds:
for vote distributions $\mathbf p$, $\mathbf p'$ and a positive integer $k$, if this method returns the solution $(a_1,\dots,a_m)$ for the problem $(\mathbf p,k )$ and the solution $(b_1,\dots,b_m)$ for the problem $(\mathbf p',k)$, then for any $i,j \in \{1,\dots,m\}$:
\[\frac{p'_i}{p'_j}\ge \frac{p_i}{p_j}\text{ implies that either }a_i'\ge a_i\text{ or }a_j'\le a_j.\]
\end{definition}

In other words, if party~$i$ increases its vote count relative to party~$j$, then either $i$ does not lose seats or $j$ does not gain seats.
We speak of a population paradox if this property is violated: a gain for party~$i$ relative to party~$j$ grants extra seats for~$j$ while~$i$ loses seats.
The following example shows that Frege's apportionment method suffers from the population paradox and thus violates population monotonicity.

\begin{example}
\setlength{\arraycolsep}{6pt}
Consider the following two scenarios with three parties ($a$, $b$, and $c$) and three seats ($k=3$):
In the first, we have 
$\mathbf p=(\frac{8}{20},\frac{3}{20},\frac{9}{20})$,
for which Frege's apportionment method yields:
\[
\begin{array}{crrrc}\toprule
\mathrm{time}&
\mathwordbox{a}{15.0}	&
\mathwordbox{b}{15.0}	&
\mathwordbox{c}{15.0}	&
\mathrm{representative}\\\midrule
1 & \frac{8}{20} & \frac{3}{20} & \mathbf{\frac{9}{20}} & c\\
2 & \mathbf{\frac{16}{20}} & \frac{6}{20} & -\frac{2}{20} & a\\
3 & \frac{4}{20} & \mathbf{\frac{9}{20}} & \frac{7}{20} & b\\
\bottomrule
\end{array}
\]
In the second scenario, we have 
$\mathbf{p}'=(\frac{5}{20},\frac{4}{20},\frac{11}{20})$.
\[
\begin{array}{crrrc}\toprule
\mathrm{time}&
\mathwordbox{a}{15.0}	&
\mathwordbox{b}{15.0}	&
\mathwordbox{c}{15.0}	&
\mathrm{representative}\\\midrule
1 & \frac{5}{20} & \frac{4}{20} & \mathbf{\frac{11}{20}} & c\\
2 & \mathbf{\frac{10}{20}} & \frac{8}{20} & \frac{2}{20} & a\\
3 & -\frac{5}{20} & \frac{12}{20} & \mathbf{\frac{13}{20}} & c\\
\bottomrule
\end{array}
\]
Now consider parties $b$ and $c$.
We have 
\[\frac{p_b}{p_c} = \frac{1}{3}< \frac{4}{11} = \frac{p'_b}{p'_c},\]
\ie, the relative strength of $b$ over $c$ increases, but $b$ loses a seat while $c$ gains one.
\end{example}

Theorem~\ref{thm:apportionment-axioms} summarizes our findings in this section.
\begin{theorem}
\label{thm:apportionment-axioms}
Frege's apportionment method satisfies house monotonicity, up\-per quo\-ta, and quota for $m\in\{2,3\}$, but fails lower quota for $m\ge 4$ and population monotonicity.
\end{theorem}

This theorem shows in particular that Frege’s apportionment method
is not a divisor method,
as divisor methods satisfy---and can even be uniquely characterized by---population monotonicity~\citep{Balinski1982FairRepresentation:Meeting}.
Furthermore, the theorem shows an axiomatic difference with the largest remainder method (which fails house monotonicity), as well
as to the quota method (which satisfies lower quota).

\subsection{Bias}
\label{subsec:bias}

As a final aspect of apportionment methods, we consider ``bias'': does a method favor small over large parties---or vice versa?
Bias is generally more of a concern when using apportionment methods for assigning representatives to states and less so for parties. In parliamentary  elections, a bias for larger parties can support the formation of governments and disincentivize schisms of parties~\citep{rae1967political,Balinski1982FairRepresentation:Meeting}.
In contrast, when assigning representatives to states, a fair treatment of large and small states is often an essential property (e.g., in the U.S.\ House of Representatives). However, an example of strong bias (in the aforementioned sense) is the European parliament, where small countries have disproportionally many members; this is referred to as \emph{degressive proportionality} \citep{RePEc:ucp:jpolec:doi:10.1086/670380}.

To formalize ``bias'' as an axiom is difficult, as it is best described as a tendency. 
\cite{Balinski1982FairRepresentation:Meeting} formalize what it means for a divisor method to be  unbiased, but this definition does not extend to arbitrary apportionment methods.
\citet{Pukelsheim2017ProportionalRepresentation:Apportionment} provides a more general, probabilistic analysis assuming that vote distributions are distributed uniformly at random and that the house size converges to infinity. 
A third approach is to compute bias in given data sets. 
This has been done by \citet{Balinski1982FairRepresentation:Meeting}  and \citet{birkhoff1976house} based on Congressional apportionment in the USA.
It is noteworthy that all these analyses yield similar results.

Our approach is to determine bias via numerical simulations.\footnote{The Python source code to run these simulations is available [reference omitted for reasons of anonymity]
}
We employ the following simple test:
we assume five parties, each having a vote count between 1 and 1000, drawn uniformly at random.
Furthermore, we assume a house size of 100 seats.
For each apportionment method, we compute the number of votes per representative of the smallest and the largest party.
If $p_l$ and $p_s$ are the vote counts for the largest and smallest party, respectively, and $a_l$ and $a_s$ are the number of seats of the largest and smallest party,
we say that the given apportionment method favors the smaller party if
\[\frac{p_s}{a_s} < \frac{p_l}{a_l},\]
\ie, if the smaller party requires fewer votes per seat.
We computed the fraction of apportionment problems where the smaller party had this advantage based on 1.000.000 instances. A value of $50\%$ would correspond to being perfectly unbiased, as small and large parties are favored equally often.

{\setlength{\tabcolsep}{12pt}
\newcolumntype{P}[1]{>{\centering\arraybackslash}p{#1}}
\begin{table}
\begin{center}
\begin{tabular}{lcc}
\toprule  & bias & 95\% confidence interval \\\midrule
Largest Remainder & 48.5\% & (48.41\%, 48.61\%) \\ 
D'Hondt (Jefferson) & 11.9\% & (11.81\%, 11.94\%) \\ 
Adams & 87.6\% & (87.51\%, 87.64\%) \\ 
Sainte-Lagu\"e (Webster) &  48.5\% & (48.38\%, 48.58\%)\\ 
Huntington-Hill &  55.7\% & (55.59\%, 55.78\%) \\ 
Quota method &  12.7\% & (12.60\%, 12.73\%)\\ 
Frege's apportionment method &  54.5\% & (54.38\%, 54.58\%) \\ \bottomrule
\end{tabular} 
\end{center}
\caption{Bias of apportionment methods computed based on numerical simulations, along with the 95\% confidence intervals. Bias, as shown here, is the percentage of instances where the smallest party is favored over the largest party; a value of 50\% corresponds to ``no bias'' for the used set of apportionment problems. 
}
\label{tab:app-bias}
\end{table}
}

The results are shown in Table~\ref{tab:app-bias} (including 95\% confidence intervals), and can be summarized as follows.
Adams favors small parties; D'Hondt and the quota method favor large parties.
Sainte-Lagu\"e and the Largest remainder method are well-balanced, as is Huntington-Hill, but to a lesser degree.
All these findings are in alignment with previous work~\citep{Balinski1982FairRepresentation:Meeting,birkhoff1976house,Pukelsheim2017ProportionalRepresentation:Apportionment}.
Frege's apportionment method achieves a ratio of $54.5\%$, in between Sainte-Lagu\"e and Huntington-Hill, and thus can be seen as a rather unbiased apportionment method.

To sum up our findings, Frege's apportionment satisfies strong proportionality guarantees, which are not achievable in the class of divisor methods. The quota method satisfies slightly stronger proportionality guarantees (both upper and lower quota), but is biased towards large parties.
Frege's method, in contrast, shows no particular bias towards large or small parties.

\section{Conclusions and Discussion}\label{sec:disc}

In our mathematical study of  Frege's voting method we focused on the extent to which it guarantees various forms of proportionality. 
Accordingly, we ignored a number of its other features that are still worth discussing.

\subsection*{Practical applicability in political elections}
It should be noted that the proportionality guarantees of Frege's method only apply to single constituencies (when observed over time) but not to the political assemblies formed by the chosen representatives. It is thus possible that the political assembly does not at all reflect the entire electorate's current political opinion. 
This issue becomes even more dramatic if one considers the actual political decision power within such an assembly \citep[cf. the work on power indices, e.g.,][]{rae1969decision,dubey1979mathematical,felsenthal1998measurement,napel2019voting}.
In particular, it may be beneficial for a group of candidates (e.g., a party) to receive few additional votes at time $t$, so that all of them are elected at time $t+1$ and thus potentially achieve a majority in the assembly.
This paradoxical behaviour leads us to the conclusion that Frege's method and the modified Frege method are only sensible for single decisions and not so much in the broader sense for electing assemblies. 

In addition, Frege's idea is only attractive if the main concern is fairness towards candidates (in the sense that no votes are lost) and only in the absence of harmful extremist opinions (as also extremist candidates would win eventually).
This is likely to be the case in low-stake, high-frequency settings, where the long-term behaviour of a mechanism is much more important than individual decisions.
In such settings, moreover, the strong assumption, occasionally made in this paper, that the electorate is fixed and does not change their preferences, would arguably also be more reasonable.
Frege's apportionment method, as introduced in Section~\ref{sec:apportionment}, is not affected by these considerations and can be recommended in situations where its axiomatic properties appear desirable.

\subsection*{Gerrymandering}
Frege claimed in his proposal (\emph{Erl\"auterungen},~302,~[9]) that his method provides a safeguard against gerrymandering (``\emph{Wahlkreisgeometrie}''), that is, the strategic districting by a political party for electoral gain, an iniquity that infamously pervades representative systems based on first-past-the-post methods for electing representatives (see, e.g., \citealt{ricca2013political}). The validity of Frege's claim, however, much
depends on the exact assumptions that are being made and accordingly warrants a careful analysis. 

The effectiveness of gerrymandering obviously rests upon the possibility of affecting the proportional support of a party in constituencies.
However, if the constituencies are of equal size and the combined electorate of all constituencies is fixed, redistricting will not affect the \emph{sum} of these proportions, no matter how clever the gerrymander.
Thus, if redistricting can only be performed once, Theorem~\ref{thm:frege_asymp} shows that, as time goes to infinity,
the number of times each candidate for a party is elected in her respective constituency will be in accordance with this proportion.
It would thus follow that Frege's method does indeed contravene the designs of gerrymanderers. 
The argument can be generalized to electorates that are not necessarily fixed but still comply with the convergence conditions as in the second part of Theorem~\ref{thm:frege_asymp}.

A number of caveats, however, are in place regarding the sweep of this argument.
First, Frege's voting method is based on the plurality rule, and as such it is still susceptible to gerrymandering for gain in the short run if an individual election in a constituency is
seen as a singular event (ignoring past and future elections) or as particularly important.
Second, even when considering the temporal nature of Frege's method, incidental gerrymandering may be successful in achieving short-term benefits without harming the candidates' long-term chances. This point also relates to the question of power distribution within an assembly, as discussed above.
Third, Frege argues that constituencies should be kept of a similar size and remain largely unchanged over time (\emph{Vorbemerkung}, 197,~[1]). This demand by itself excludes some forms of gerrymandering but cannot be seen as a (mathematical) guarantee of his voting method.
Therefore, Frege's method certainly prevents or hinders certain forms of gerrymandering, but to which extent and under which assumptions is a question we leave for future research.
Furthermore, it would be interesting to investigate whether the modified Frege method guarantees a better protection against gerrymandering.

\subsection*{Choice candidates and delegation}
Among all candidates in a constituency, Frege proposes to distinguish so-called \emph{choice candidates} (``\emph{Erlesene}''), the twenty-five candidates in a constituency with maximal electoral backing, among which the representative will be chosen.
So as to ensure that no votes are lost, Frege also provisioned for a delegation mechanism, in which non-choice candidates or deceased choice candidates can transfer the votes cast on them to one of the (living) choice candidates. 
The exact social choice theoretic ramifications of this delegation mechanism are left as a topic for future research.
Furthermore, this mechanism could be compared with modern proposals for vote delegation~\citep{alger2006voting,green2015direct,blum2016liquid}.

\subsection*{Negative aggregate scores and strategic voting}
We have seen how an increasing cost of winning undermines the proportionality guarantees of Frege's method until this cost stabilizes (after a potentially very long time).
A further disadvantage of the original method is that long-serving candidates tend to have high scores, which makes the entry of new candidates difficult, even if they have a strong public support.
In contrast, the modified Frege method has a constant cost of winning and, as we have seen, stronger proportionality guarantees.
However, here the aggregate scores of the candidates can be negative, which also leads to negative consequences.
The possibility of negative scores renders the modified Frege method vulnerable to the following type of manipulation. Once a candidate has a negative score, it is advantageous for this candidate to retract his or her candidacy in favor of a like-minded person (in social choice terminology a so-called clone), who then starts with a higher aggregate score of~$0$.
A complete study of the manipulability of Frege's voting method and the modified Frege method is subject to future research.

\subsection*{Transition from plurality to proportionality}
The plurality rule performs very badly when it comes to proportionality in the long run. 
If the electorate is assumed to be fixed, it will always elect the same candidate! 
Frege's method has much better proportionality properties, and the modified Frege method even better ones still. 
In an important and interesting sense, Frege's method can be seen as a gradual transition between a system based on plurality towards a system based on the modified Frege method, as the cost of winning is increased until it stabilizes at the size of the electorate and henceforward behaves like the modified Frege method.

\subsection*{Outlook and research directions}
The temporal or dynamic aspect of the Frege methods distinguishes them from most other voting methods that have been considered in the literature. 
As such they can also take into account changing electorates and changing opinions among the electorate.
Yet, the dynamic Frege methods still rely on the plurality rule in that plurality ballots are used and, for any election at any one single time, they simply select the candidate with the highest aggregate score.
Seen this way,  the Frege methods could easily be varied upon by considering other `static' social choice rules instead of the plurality rule, thus defining a new class of dynamic voting rules that can be studied from a social choice perspective in their own right. 
An obvious variation, for instance, would be to assume that the voters' ballots specify complete preference orders over the candidates. This would allow the computation of Borda scores.
The candidates then aggregate their respective Borda scores over time in a similar way as they aggregate plurality scores for the Frege methods.
At each election the candidate with the highest Borda score could then be chosen as representative and subsequently incur a certain cost of winning yet to be defined. 
In order to investigate this class of dynamic voting rules in a systematic and principled fashion, one may want to define  axioms that are specific to the temporal setting, like variable quota axioms. In particular it would be interesting to see if there is a dynamic voting rule that satisfies both variable upper and lower quota.

\section*{Acknowledgements}
Paul Harrenstein was supported by the European Research Council via ERC Advanced Investigator Grant 291528 (``RACE'') and ERC Grant 639945 (``ACCORD'') at Oxford, as well as by the Alan Turing Institute in London.
The financial support of Marie-Louise Lackner by the Austrian Federal Ministry for Digital and Economic Affairs and the National Foundation for Research, Technology and Development is gratefully acknowledged. 
Martin Lackner was supported by the Austrian Science Fund FWF, grant P31890, and ERC grant 639945.

\bibliographystyle{abbrvnat}

\newpage
\appendix
\section{Proofs of Section~\ref{sec:prop-orig}}
\label{sec:app1}

\textbf{Lemma~\ref{lem:aggaggscore_leq_nm}.}\ \ 
\lemaggaggscoreleqnm

\begin{proof}
Let us establish a recursive definition for $a(t)$, starting with
\[
a(1)=\sum_{k\in\candidates}\aggscore_k^1 =\sum_{k\in\candidates}\pluralityscore_k^1=n.
\]
Let $j^*=\representative(t)$.	Then the following equalities hold:
\begin{align*}
		a(t+1) = \sum_{k\in C}\aggscore_k^{t+1}
		&	\mathwordbox{=}{===}	\aggscore^{t}_{j^*}+\pluralityscore_{j^*}^{t+1}-\left\lfloor\frac{1}{m}\sum_{k\in C}\aggscore^{t}_k\right\rfloor+\sum_{k\neq j^*}(\aggscore_k^{t}+\pluralityscore_k^{t+1})\\
		&	\mathwordbox{=}{===}
							\sum_{k}\aggscore_k^{t}
							+ \sum_{k}\pluralityscore_k^{t+1}
							-\left\lfloor\frac{1}{m}\sum_{k\in C}\aggscore^{t}_k\right\rfloor\\
		&	\mathwordbox{=}{===}	a(t)
							+ n
							-\left\lfloor\frac{1}{m}a(t)\right\rfloor \text{for } t\ge 1.
	\end{align*}

Now, let us start by proving that there exists a $t_0 \in \mathbb{N}$ such that $a(t) \ge n \cdot m$ for all $t \ge t_0$.
For this purpose, let us consider the simpler recursion $b(1)=n$ and $b(t+1)= b(t) +n -{{\frac{1}{m}\cdot b(t)}}$. This recursion has the solution $b(t) = nm\left(1-\left(\frac{m-1}{m}\right)^t\right)$. Note that $b(t)$ converges to $nm$ for $t\to \infty$. Furthermore, it holds that $b(t)\le a(t)$ (this can easily be shown by induction). Since $b(t)$ converges to $nm$ and $a(t)$ is integer-valued, there has to be a point in time $t_1$ such that $a(t)\ge nm$ for all $t \ge t_1$.

Let $t_0 \ge 1$ be the smallest possible choice for $t_1$, that is, $t_0$ is chosen such that $a(t_0) \ge nm$ and  $a(t_0-1) < nm$. We want to show that $a(t_0)=nm$. Let $i \in \mathbb{N}$ be such that $a(t_0 -1)=nm -i$. Then, using the recursion for $a(t)$, we obtain the following:
\begin{align*}
a(t_0) 	&=  a(t_0-1) + n - \left\lfloor\frac{a(t_0-1)}{m}\right\rfloor \\
		& =  nm - i + n - \left\lfloor\frac{nm -i}{m}\right\rfloor \\
		& = nm - i + n - n - \left\lfloor-\frac{i}{m}\right\rfloor \\
		& = nm - i + \left\lceil\frac{i}{m}\right\rceil \\
		& \le nm,
\end{align*}
since $- i + \left\lceil\frac{i}{m}\right\rceil \le 0$.
It follows that $a(t_0)=nm$.

Now, let us prove that $a(t)=nm$ for all $t \ge t_0$. Using the recursion for $a(t)$, we have $a(t_0 +1)=nm + n - \left\lfloor \frac{nm}{m}\right\rfloor=nm$. By induction $a(t)=nm$ for all $t\ge t_0$.

Let us finally show that the function $a(t)$ is monotonically increasing. From what was proven so far, we now that $0 \le a(t) \le nm$ and thus $0 \le \left\lfloor \frac{a(t)}{m}\right\rfloor \le n$ for all $t \ge 1$. Thus, we have for all $t \ge 1$:
\begin{align*}
a(t+1) & = a(t)	+ n  -\left\lfloor\frac{1}{m}a(t)\right\rfloor \ge a(t) +n -n =a(t).
\end{align*}
\end{proof}

\medskip
\noindent\textbf{Theorem~\ref{thm:frege_asymp}.}\ \
\thmfregeasymp

Let us first prove a technical lemma, which is required in the proof of Theorem~\ref{thm:frege_asymp}:
\begin{lemma}
Assume that the number of voters $n$ is fixed. For every candidate $j \in C$ there exists a positive integer $c_j$ such that for all $t \ge t_0$
\begin{equation}
\aggscore_j^{t+1}=\sum_{s=1}^{t+1}\pluralityscore_j^s - c_j -n\cdot  \left(\rho_j(t)-\rho_j(t_0)\right), \label{eqn:aggscore}
\end{equation}
where $t_0 \ge 1$ is such that $\sum_{k\in\candidates}\aggscore_k^{t_0}=nm$.
\label{lem:eqn_aggscore}
\end{lemma}

\begin{proof}
Let $t_0 \ge 1$ be such that $\sum_{k\in\candidates}\aggscore_k^{t_0}=nm$ (the existence of such a $t_0$ was proven in Lemma~\ref{lem:aggaggscore_leq_nm}), and let $c_j\ge 0$ be such that
\[
\aggscore_j^{t_0+1}=\sum_{s=1}^{t_0+1}\pluralityscore_j^s-c_{j}.
\]
We shall prove equation~\eqref{eqn:aggscore} by induction over $t \ge t_0$.

For the induction start, let us consider time $t_0$:
\begin{align*}
\aggscore_j^{t_0+1}& =\sum_{s=1}^{t_0+1}\pluralityscore_j^s-c_{j} =\sum_{s=1}^{t_0+1}\pluralityscore_j^s-c_j- n \cdot \left(\rho_j(t_0)-\rho_j(t_0)\right).
\end{align*}

For the induction step, we distinguish whether $\representative(t+1)=j$ or not. We may assume that $\aggscore_j^{t+1}=\sum_{s=1}^{t+1}\pluralityscore_j^s - c_j -\left(\rho_j(t)-\rho_j(t_0)\right)\cdot n$.
In case $\representative(t+1)=j$, we have:
\allowdisplaybreaks
\begin{align*}\allowdisplaybreaks
\aggscore_j^{t+2}& = \aggscore_j^{t+1} + \pluralityscore_j^{t+2} - \left\lfloor\frac{1}{m}\sum_{k\in C}\aggscore^{t+1}_k\right\rfloor \\
& = \sum_{s=1}^{t+1}\pluralityscore_j^s - c_j -\left(\rho_j(t)-\rho_j(t_0)\right)\cdot n +  \pluralityscore_j^{t+2} -n  \\
& =\sum_{s=1}^{t+2}\pluralityscore_j^s-c_j- \left(\rho_j(t)+1-\rho_j(t_0)\right) \cdot n \\
& =\sum_{s=1}^{t+2}\pluralityscore_j^s-c_j- \left(\rho_j(t+1)-\rho_j(t_0)\right) \cdot n.
\end{align*}
In case $\representative(t+1)\neq j$, we have:
\begin{align*}\allowdisplaybreaks
\aggscore_j^{t+2}& = \aggscore_j^{t+1} + \pluralityscore_j^{t+2}  \\
& = \sum_{s=1}^{t+1}\pluralityscore_j^s - c_j -\left(\rho_j(t)-\rho_j(t_0)\right)\cdot n +  \pluralityscore_j^{t+2}  \\
& =\sum_{s=1}^{t+2}\pluralityscore_j^s-c_j- \left(\rho_j(t)-\rho_j(t_0)\right) \cdot n \\
& =\sum_{s=1}^{t+2}\pluralityscore_j^s-c_j- \left(\rho_j(t+1)-\rho_j(t_0)\right) \cdot n,
\end{align*}
which concludes the induction step.
\end{proof}

\begin{proof}[Proof of Theorem~\ref{thm:frege_asymp}]
Assuming a fixed electorate, equation~\eqref{eqn:aggscore} becomes:
\begin{equation}
\aggscore_j^{t+1}=(t+1)\cdot \pluralityscore_j - c_j -\left(\rho_j(t)-\rho_j(t_0)\right)\cdot n
\label{eqn:aggscore_fixed_electorate}
\end{equation}
or, equivalently,
\[
\frac{\rho_j(t)}{t}=\frac{\pi_j}{n}+\frac{1}{t}\left(\frac{\pi_j}{n}+ \rho_j(t_0) -\frac{c_j}{n}-\frac{\aggscore_j^{t+1}}{n}\right)
\]
To see that the expression $\frac{1}{t}(\dots)$ converges to 0, note the following:
First, $0 \le \frac{\pi_j}{n} \le 1$ and  $0 \le \rho_j(t_0) \le t_0$. Moreover,
$0 \le c_j \le \rho_j(t_0) \cdot n\le t_0\cdot n$ and $0\le \aggscore_j^{t+1} \le n \cdot m$. Thus, the term in brackets is bounded from below and above. We obtain:
\[
\lim_{t \to \infty} \frac{\rho_j(t)}{t} = \frac{\pi_j}{n}.
\]
Similarly, if the electorate is not fixed but the mean plurality scores converge, that is, $\lim_{t \to \infty} \nicefrac 1 t \cdot \sum_{s=1}^t\pluralityscore_j^s = \pluralityscore_j^*$ for all $j\in\candidates$, we have:
\[
\frac{\rho_j(t)}{t}=\frac{\sum_{s=1}^t\pluralityscore_j^s}{t \cdot n}+\frac{1}{t}\left(\frac{\pluralityscore_j^{t+1}}{n} + \rho_j(t_0) -\frac{c_j}{n}-\frac{\aggscore_j^{t+1}}{n}\right)
\]
and thus
\[
\lim_{t \to \infty} \frac{\rho_j(t)}{t} = \lim_{t \to \infty} \frac{\sum_{s=1}^t\pluralityscore_j^s}{t \cdot n} = \frac{\pluralityscore_j^*}{n}.
\]
\end{proof}

The following proposition shows that, for a fixed electorate, proportionality is not only guaranteed in the limit but eventually also within (finite) intervals.

\medskip
\noindent\textbf{Proposition~\ref{prop:periodic}.}\ \ 
\propperiodic

\begin{proof}
From Lemma~\ref{lem:aggaggscore_leq_nm} we know that  $0 \le \aggscore_k^t \le nm$ for all $k\in C$ and $t\geq 1$.
Thus the tuple $(\aggscore_1^t, \aggscore_2^t, \ldots, \aggscore_m^t)$ can only take finitely many values. Therefore there must be a time $t^* \ge 1$ and an integer $P$ such that
\[
(\aggscore_1^{t^*}, \aggscore_2^{t^*}, \ldots, \aggscore_m^{t^*}) = (\aggscore_1^{t^*+P}, \aggscore_2^{t^*+P}, \ldots, \aggscore_m^{t^*+P}).
\]
Given these two values $t^*$ and $P$, it clearly also holds that
\begin{align*}
(\aggscore_1^{t^* + k}, \aggscore_2^{t^*+k}, \ldots, \aggscore_m^{t^*+k}) & = (\aggscore_1^{t^*+P+k}, \aggscore_2^{t^*+P+k}, \ldots, \aggscore_m^{t^*+P+k}) \text{ for } k \in \mathbb{N} \text{ and thus } \\
(\aggscore_1^{t}, \aggscore_2^{t}, \ldots, \aggscore_m^{t}) & = (\aggscore_1^{t+P}, \aggscore_2^{t+P}, \ldots, \aggscore_m^{t+P}) \text{ for all } t \ge t^*.
\end{align*}
From equation~\eqref{eqn:aggscore} in Lemma~\ref{lem:eqn_aggscore} it follows that the following holds for $t \ge t^*-1$:
\begin{align*}
 \rho_j(t+P) - \rho_j(t_0) & = \frac{1}{n} \cdot \left( \pi_j \cdot (t+P +1) - c_j - \aggscore_j^{t+P+1}\right) \\
\text{ and } \quad \quad \, \rho_j(t) - \rho_j(t_0) & = \frac{1}{n} \cdot \left( \pi_j \cdot (t+1) - c_j - \aggscore_j^{t+1}\right).
\end{align*}
Thus
\[
\rho_j(t+P) - \rho_j(t)  = \frac{1}{n} \cdot \left( \pi_j \cdot P - (\aggscore_j^{t+P+1} - \aggscore_j^{t+1})\right)
 = \frac{\pi_j \cdot P}{n},
\]
which concludes the proof.
\end{proof}

\section{Proofs of Section~\ref{sec:stronger}}
\label{sec:app2}

\noindent\textbf{Theorem~\ref{thm:mod_frege_asymp}.}\ \ 
\thmmodfregeasymp

\newcommand{\lemmaaggaggscoreis}{For the modified Frege method, we have $\sum_{j\in C}s^t_j=1$ and $- 1 < s^t_j-p_j^t$ for all $j \in C$ and $t\ge 1$.}

\medskip
Let us first prove a technical lemma, which will yield the desired proportionality results.
\begin{lemma}\lemmaaggaggscoreis\label{lemma:aggaggscore_is_1}
\end{lemma}

\begin{proof}
	The proof of the first statement is by induction on~$t$.
	For the basis $t=1$, we immediately have
	\[
		\sum_{j\in C}s^1_j
		\mathwordbox{=}{==}\sum_{j\in C}p^1_j
		\mathwordbox{=}{==}\sum_{j\in C}\frac{\pluralityscore^1_j}{n_1}
		\mathwordbox{=}{==}1
	\text.
	\]

	For the induction step, we assume $\sum_{j\in C}s^t_j=1$.
	Let $j^*=\representative(t)$. Then,
	\begin{align*}
		\sum_{j\in C}s^{t+1}_j
		&	\mathwordbox{=}{===}	\sum_{k\neq\winner}\left(s^t_k+p_k^{t+1}\right)+s^t_{j^*}+p_{j^*}^{t+1}-1\\
		&	\mathwordbox{=}{===}	\sum_{j\in C}s_j^t+\sum_{j\in C}p_j^{t+1}-1 = 1	+	1	-1 = 1.
	\end{align*}
The proof of the second statement is by induction on~$t$ as well.
	For the basis $t=1$, we immediately have
	\[
	s_j^1 -p_j^1 = p_j^1 -p_j^1 = 0 > -1\text.
	\]
	For the induction step, we assume $-1 < s_j^t -p_j^t$. We distinguish two cases. First, if $j \neq \representative(t)$, we have:
	\[
	s_j^{t+1} -p_j^{t+1}= s_j^t + p_j^{t+1} - p_j^{t+1} = s_j^t \ge s_j^t -p_j^t > -1.
	\]
	Second, if $j = \representative(t)$, it has to hold that $s_j^t \ge \frac{\sum_{j \in C} s_j^t}{m} = \frac{1}{m}$ (in order for $j$ to win). Thus we have:
	\[
	s_j^{t+1} - p_j^{t+1} = s_j^t + p_j^{t+1} -1  - p_j^{t+1} = s_j^t -1\ge \frac{1}{m} - 1 > -1.
	\]
\end{proof}

\begin{proof}[Proof of Theorem~\ref{thm:mod_frege_asymp}]
Note that the following holds for the modified Frege method:
\begin{equation}
s_j^{t+1}=\sum_{s=1}^{t+1}p_j^s -r_j(t) \text{ for all } t \ge 1 \text{ and } j \in C.
\label{eqn:aggscore_frege_mod}
\end{equation}
This corresponds to a normalized version of equation~\eqref{eqn:aggscore} in the proof of Lemma~\ref{lem:eqn_aggscore} with $r_j(t_0)=0$ and $c_j=0$.

It follows that
\[
\frac{r_j(t)}{t}=\frac{\sum_{s=1}^{t+1}p_j^s}{t} -\frac{s_j^{t+1}}{t}.
\]
Since we know from Lemma~\ref{lemma:aggaggscore_is_1} that $-1 \le s_j^{t+1} \le 1$,
both asymptotic results follow immediately.
\end{proof}

We proceed with another technical lemma, which is useful for proving Theorems~\ref{thm:upperquota} and~\ref{thm:lower-quota}.

\begin{lemma} For every $t\ge 1$ and $j\in C$:
		\[
		-1
		\mathwordbox{<}{===}  \sum_{s=1}^{t}p_j^s-r_j(t)
		\mathwordbox{\le}{===} \frac{m-1}{2}.
		\]
	 The second inequality  is strict for $m\ge 3$.
	 \label{lem:bounds_for_quota}
\end{lemma}

\begin{proof}
The lower bound of~$-1$ follows from a combination of Lemma~\ref{lemma:aggaggscore_is_1} and Theorem~\ref{thm:mod_frege_asymp}:
	By virtue of equation~\eqref{eqn:aggscore_frege_mod} in the proof of Theorem~\ref{thm:mod_frege_asymp}, the lower bound is equivalent to $-1< s_j^{t+1}-p_j^{t+1}$. Lemma~\ref{lemma:aggaggscore_is_1} ensures that this inequality holds for all $j \in C$ and all $t \ge 1$.

	To prove the upper bound of $\frac{m-1}{2}$, by virtue of equation~\eqref{eqn:aggscore_frege_mod}, it suffices to show by induction on~$t$ that
	$s_j^{t+1}- p_j^{t+1}\le \frac{m-1}{2}$ for all candidates~$j$ and~$t\ge 1$.
	For the basis, let $t=1$. First assume that $j=\representative(1)$.
	Then, observing that $0<p_j^1\le 1$,
	\[
		s^2_j-p_j^{2}
		\mathwordbox{=}{===}	s_j^1+p_j^{2}-1-p_j^{2}
		\mathwordbox{=}{===}	p_j^{1}-1
		\mathwordbox{\le}{===}	0
		\mathwordbox{<}{===}	\textstyle\frac{m-1}{2}\text,
	\]
	where the last inequality is strict as we always assume $m\ge 2$.
	Now assume that $j\neq \representative(1)$.
	Then, $p_j^{1}\le \frac{1}{2}$, as otherwise~$j$ would have been chosen as representative.
	Accordingly,
	\[
		s^2_j-p_j^{2}
		\mathwordbox{=}{===}	s_j^1+p_j^{2}-p_j^{2}
		\mathwordbox{=}{===}	p_j^{1}
		\mathwordbox{\le}{===}	\textstyle \frac{1}{2}
		\mathwordbox{\le}{===}	\textstyle\frac{m-1}{2}\text,
	\]
	where the last inequality is strict if $m\ge3$.

	We introduce the following notation: 
	we write $<^*$ to denote a weak inequality ($\le$) if $m=2$ and a strict inequality ($<$) for $m>2$.
	For the induction step, we may assume $s_j^{t+1}-p_j^{t+1}<^* \frac{m-1}{2}$ to prove that $s_j^{t+2}-p_j^{t+2}<^* \frac{m-1}{2}$.
	First assume that $j=\representative(t+1)$.
	Now the following inequalities hold.
	\begin{align*}
		s_j^{t+2}-p_j^{t+2}
		\mathwordbox{}{}\mathwordbox{=			}{===}\textstyle s_j^{t+1}+p_j^{t+2}-1-p_j^{t+2}
		\mathwordbox{}{}\mathwordbox{=			}{===}\textstyle s_j^{t+1}-1
		\mathwordbox{}{}\mathwordbox{\le 		}{===}\textstyle s_j^{t+1}-p_j^{t+1}
		\mathwordbox{}{=}\mathwordbox[l]{<^*_{i.h.}}{===}	\textstyle\frac{m-1}{2}\text.
	\end{align*}
	Now, let $j^*=\representative(t+1)$ and assume that $j\neq j^*$.
	Accordingly, $s^{t+1}_j\le s^{t+1}_{j^*}$.
	By virtue of Lemma~\ref{lemma:aggaggscore_is_1}, we have that
	\[
		 s_j^{t+1}+s_{j^*}^{t+1}+\sum_{k\in C\setminus\set{j,j^*}}s_k^{t+1}=1.
	\]
	As we saw above that $-1 < s_j^{t+1}-p_j^{t+1}$, we also have the following:
	\begin{align*}
		s_j^{t+1}+s_{j^*}^{t+1}
		&	\mathwordbox{=	}{===}\textstyle	1-\sum_{k\in C\setminus\set{j,j^*}}s_k^{t+1}					\\
		&	\mathwordbox{\le}{===}\textstyle		1-\sum_{k\in C\setminus\set{j,j^*}}(s_k^{t+1}-p_k^{t+1})\\
		&	\mathwordbox{<^*}{===}\textstyle 	1+m-2		\\
		&	\mathwordbox{=	}{===}\textstyle	\textstyle m-1
		\text{.}
	\end{align*}
	As $s^{t+1}_j\le s^{t+1}_{j^*}$, it follows that
		$s_j^{t+1}<^*\textstyle\frac{m-1}{2}$.
		Finally, since $r_j(t+1)=r_j(t)$,
		\begin{align*}
       s_j^{t+2}-p_j^{t+2}=
       \sum_{s=1}^{t+2}p_j^s -r_j(t+1)-p_j^{t+2}=
       \sum_{s=1}^{t+1}p_j^s -r_j(t)=       
       s_j^{t+1}<^*
	   {\textstyle \frac{m-1}{2},}
		\end{align*}
which concludes the induction.
\end{proof}

\medskip
\noindent\textbf{Theorem~\ref{thm:upperquota}.}\ \ 
\thmupperquota

\begin{proof}
By Lemma~\ref{lem:bounds_for_quota}, it holds that
\[\sum_{s=1}^{t}p_j^s-r_j(t) > -1,\]
and consequently
\[\norepr_j(t) < \sum_{s=1}^{t}p_j^s+1.\]
This is equivalent to
\[\norepr_j(t) \le \left\lceil\sum_{s=1}^{t}p_j^s\right\rceil,\]
which is exactly the condition for variable upper quota.
\end{proof}

\medskip
\noindent\textbf{Theorem~\ref{thm:lower-quota}.}\ \
\thmlowerquota

\begin{proof}
By Lemma~\ref{lem:bounds_for_quota}, for $m\ge 3$ it holds that,
\[\sum_{s=1}^{t}p_j^s-r_j(t) < \frac{m-1}{2} 
\]
and consequently
\[\norepr_j(t) > \sum_{s=1}^{t}p_j^s - \frac{m-1}{2}.\]
This implies
\[\norepr_j(t) \ge \left\lfloor\sum_{s=1}^{t}p_j^s\right\rfloor- \left\lceil\frac{m-3}{2}\right\rceil.\]
For $m=3$, the last inequality becomes
\[\norepr_j(t) \ge \left\lfloor\sum_{s=1}^{t}p_j^s\right\rfloor\]
and thus lower quota is fulfilled.

For $m=2$,
variable lower quota follows from variable upper quota.
Let $C=\{a,b\}$ and $\sum_{s=1}^{t}p_j^s=x_j$ for $j\in\{a,b\}$.
Thus, $x_a + x_b = t$.
Towards a contradiction, assume without loss of generality that $r_a(t)<\lfloor x_a\rfloor$, that is, candidate~$a$'s lower quota is violated.
Then:
\begin{align*}
r_b(t) &= t - r_a(t)
> t - \lfloor x_a\rfloor
= t - \lfloor t-x_b\rfloor
= \lceil x_b\rceil,
\end{align*}
which is in contradiction to variable upper quota for candidate~$b$.
\end{proof}

\end{document}